\documentclass[conference]{IEEEtran}
\usepackage{cite}
\usepackage{amsmath,amssymb,amsfonts,amsthm,mathrsfs}
\usepackage{multirow}
\usepackage{bbm}
\usepackage{algorithmic}
\usepackage{caption}
\usepackage{graphicx}
\usepackage{tikz} 
\usetikzlibrary{automata, positioning, arrows}
\usepackage{textcomp}
\usepackage{xcolor}
\usepackage{manfnt}
\usepackage{mathtools}
\usepackage{fixmath}
\usepackage{pgf}
\usepackage{xcolor}
\usetikzlibrary{positioning,arrows}
\usepackage{balance} 

\def\BibTeX{{\rm B\kern-.05em{\sc i\kern-.025em b}\kern-.08em
    T\kern-.1667em\lower.7ex\hbox{E}\kern-.125emX}}
    
\newcommand{\N}{\mathbb{N}}

\newcommand{\R}{\mathbb{R}}
\newcommand{\Q}{\mathbb{Q}}

\newcommand{\sA}{\mathcal{A}}

\newcommand{\sP}{\mathcal{P}}

\newcommand{\sS}{\mathcal{S}}
\newcommand{\sU}{\mathcal{U}}
\newcommand{\sX}{\mathcal{X}}
\newcommand{\sY}{\mathcal{Y}}

\newcommand{\sW}{\mathcal{W}}

\newcommand{\fT}{\mathfrak{T}}

\theoremstyle{plain}
\newtheorem{thm}{Theorem}
\newtheorem{lem}{Lemma}

\newtheorem{cor}{Corollary}

\theoremstyle{defn}
\newtheorem{defn}{Definition}

\theoremstyle{rem}
\newtheorem{rem}{Remark}

\tikzstyle{block} = [draw, rectangle, 
  minimum height=3em, minimum width=4em]

\newcommand{\tcr}[1]{\textcolor{red}{#1}}
\newcommand{\tco}[1]{\textcolor{orange}{#1}}

\renewcommand{\tcr}[1]{#1}
\renewcommand{\tco}[1]{#1}

\begin{document}

\title{Capacity of Finite State Channels with Feedback: Algorithmic and Optimization Theoretic Properties}

\author{
	\IEEEauthorblockN{Andrea Grigorescu\IEEEauthorrefmark{1}, Holger Boche\IEEEauthorrefmark{1}\IEEEauthorrefmark{2}, Rafael F. Schaefer\IEEEauthorrefmark{3}, and H. Vincent Poor\IEEEauthorrefmark{4}\\[1.5ex]}
	\IEEEauthorblockA{\small \IEEEauthorrefmark{1} Chair of Theoretical Information Technology, Technical University of Munich\\
		\IEEEauthorrefmark{2} BMBF Research Hub 6G-life\\
		Excellence Cluster Cyber Security in the Age of Large-Scale Adversaries (CASA), Ruhr University Bochum\\
		email: \texttt{\{andrea.grigorescu, boche\}@tum.de}\\[1ex]
		\IEEEauthorrefmark{3} Chair of Communications Engineering and Security, University of Siegen\\
		Center for Sensor Systems (ZESS), University of Siegen\\
		email: \texttt{rafael.schaefer@uni-siegen.de}\\[1ex]
		\IEEEauthorrefmark{4} Department of Electrical and Computer Engineering, Princeton University\\
		email: \texttt{poor@princeton.edu}
	}
	}

\maketitle

\begin{abstract}
The capacity of finite state channels (FSCs) with feedback has been shown to be a limit of a sequence of multi-letter expressions. Despite many efforts, a closed-form single-letter capacity characterization is unknown to date. In this paper, the feedback capacity is studied from a fundamental algorithmic point of view by addressing the question of whether or not the capacity can be algorithmically computed. To this aim, the concept of Turing machines is used, which provides fundamental performance limits of digital computers. It is shown that the feedback capacity of FSCs is not Banach-Mazur computable and therefore not Borel-Turing computable. As a consequence, it is shown that either achievability or converse is not Banach-Mazur computable, which means that there are computable FSCs for which it is impossible to find computable tight upper and lower bounds. Furthermore, it is shown that the feedback capacity cannot be characterized as the maximization of a finite-letter formula of entropic quantities. 
\end{abstract}

\footnotetext[1]{\tco{This work of H. Boche was supported in part by the German Federal Ministry of Education and Research (BMBF) within the national initiative on \tcr{6G Communication Systems through the research hub \emph{6G-life} under Grant 16KISK001K, within the national initiative for \emph{Post Shannon Communication (NewCom)} under Grant 16KIS1003K, and the project \emph{Hardware Platforms and Computing Models for Neuromorphic Computing (NeuroCM)} under Grant 16ME0442. It has further received funding by the Bavarian Ministry of Economic Affairs, Regional Development and Energy as part of the project \emph{6G Future Lab Bavaria}} as well as in part by the German Research Foundation (DFG) within Germany’s Excellence Strategy EXC-2092 -- 390781972. This work of R. F. Schaefer was supported in part by the BMBF within NewCom under Grant 16KIS1004 and in part by the DFG under Grant SCHA 1944/6-1. This work of H. V. Poor was supported by the U.S. National Science Foundation under Grant \tcr{ CCF-1908308.}}}

\section{Introduction}
\label{sec:introduction}

Finite state channels (FSCs) model channels with memory where the channel output depends not only on the current channel input but also on the underlying channel state. The channel state allows the channel output to implicitly depend on previous channel inputs and outputs. FSCs are of significant interest as they allow to model certain types of channel variations appearing in wireless communications including, e.g., flat fading and intersymbol interference (ISI) \cite{gallager1968information}. 

In information theory, it has been always of interest to compute the capacity of channels or channel reliability functions. In 1967, techniques for constructing simple upper and lower bounds for channel reliability functions were introduced in \cite{shannon1967lower}. These techniques were developed with the goal of computing those bounds on digital computers.
 In 1972 an algorithm to compute the capacity of arbitrary discrete memoryless channels (DMC) was independently presented in \cite{arimoto1972algorithm} and \cite{blahut1972computation}. In \cite{blahut1972computation}, an analogous algorithm was proposed to compute the rate distortion of lossy source compression. In general, the capacity is \tco{usually} given by mutual information expression. \tco{Note} that even for the  binary symmetric channel \tco{(BSC)} with rational crossover probability, i.e., \tco{$\epsilon\in(0,\frac{1}{2})\cap\Q$}, the capacity is a transcendental number. Hence, a precise calculation is not possible, since the calculation has to stop after a finite number of computation steps. Only a suitable approximation of it can be calculated.
 
After the works in \cite{shannon1967lower,arimoto1972algorithm,blahut1972computation,csiszar1974computation}  were published, it became increasingly popular in information and communication theory to simulate performance measures of communication systems on digital computers\tco{; in particular for multi-user communication scenarios.} In multi-user information theory, the progress has been rather limited. Therefore, network simulations on digital and high performance computers became a widely used method for the design of practical systems. For a critical discussion of this trend, we refer to \cite{ephremides1998information}. Network simulation plays a crucial role in the design and standardization of communication networks here also.

Determining the capacity of FSCs is a very difficult task. The trapdoor channel, for example, is simple to describe, however its capacity is still an open problem. For now, only a lower bound \cite{kobayashi2002input} and an upper bound given by the feedback capacity \cite{permuter2008capacity} are known.  For general FSCs, a finite-letter characterization of the capacity in closed form is not known to date; only a general formula based on the inf-information rate has been established in \cite{verdu1994general}. Moreover, in \cite{boche2020shannon} it was shown that the FSC capacity is not a computable function.
 
In \cite{shannon1956zero}, it was shown that feedback does not increase the capacity for DMC. However the zero error capacity for a channel with feedback might be greater in some cases, while there is still no closed formula for the zero error capacity without feedback so far. The feedback capacity of \tcr{Markov channels} was studied in \cite{tatikonda2000control}. The capacity of general FSCs with feedback was studied in \cite{tatikonda2008capacity,permuter2009finite,dabora2013capacity}. Only a multi-letter characterization of the capacity is known to date.

In recent years there has been a growing interest on computing the capacity function for FSCs with feedback. The feedback capacity was first formulated as a dynamic program for Markov channels without ISI in \cite{tatikonda2000control} and for a subclass of Markov channels in \cite{chen2005capacity}. Modeling the feedback capacity as a dynamic program and implementing algorithms to solve the Bellman equation has also been used to compute the feedback capacity of the  trapdoor channel \cite{permuter2008capacity}, the binary Ising channel \cite{elishco2014capacity}, the \tco{input-constrained BSC} \cite{sabag2015feedback} and the \tco{input-constrained} binary erasure channel \cite{sabag2015feedback2}. In \cite{aharoni2019computing}, reinforcement learning (RL) algorithms have been proposed to estimate the feedback capacity of a class of unifilar FSCs. The algorithms mentioned above all depend on the specific channel under consideration. In this paper, we address the question of whether or not an algorithm exists that takes as input an \emph{arbitrary} FSC with feedback and computes its capacity.

We are interested in the existence of ``\emph{simple}'' capacity expressions and whether or not such capacity expressions for FSCs with feedback with fixed and finite alphabets are algorithmically computable. In information theory, some performance functions are implicitly assumed to be computable. In particular, capacity expressions with entropic quantities are usually assumed to be algorithmically computable. 

To address algorithmic computability, we use the concept of a \emph{Turing machine} \cite{turing1936computable,turing1938computable,weihrauch2000computable}, which is a mathematical model of an abstract machine that manipulates symbols on a strip of tape according to certain given rules. It can simulate any given algorithm and therefore provides a simple and very powerful model of computation. Turing machines have no limitations on computational complexity, computing capacity or storage, and execute programs completely error-free. Accordingly, they provide fundamental performance limits for today's digital computers. Turing machines account for all those problems and tasks that are algorithmically computable on a classical (i.e., non-quantum) machine. They are further equivalent to the von Neumann-architecture without hardware limitations and the theory of recursive functions \cite{godel1930vollstandigkeit,godel1934undecidable,kleene1952introduction,minsky1961recursive,avigad2014computability}. 

The computability of the capacity of FSCs has been studied in \cite{elkouss2018memory}, where it was shown that the capacity of FSCs is not \tcr{Borel-Turing} computable if the input and state alphabets $\sX$ and $\sS$ satisfy $|\sX|\geq10$ and $|\sS|\geq62$. In \cite{boche2020shannon}, it was shown that the capacity of FSCs is in general not \tcr{Borel-Turing} computable even for the smallest non-trivial case, i.e., $|\sX|\geq 2$, $|\sY|\geq 2$, and $|\sS|\geq 2$. For $|\sS|=1$, the channel is a DMC and the capacity is given by Shannon's single-letter formula. The capacity of a DMC is \tcr{Borel-Turing} computable.

This paper addresses the general question of whether or not a finite-letter characterization of the capacity of FSCs with feedback exists at all and whether or not the feedback capacity of FSCs is algorithmically computable. For FSCs with $|\sX | \geq 2, |\sY| \geq 2$, and $|\sS| \geq 2$, we show that the feedback capacity is not Banach-Mazur computable and therefore also not Borel-Turing computable. We show that if the capacity of FSCs with feedback had been computable, then it could yield a solution for the halting problem. The halting problem is a decision problem in computability theory, which has been proven to be undecidable \cite{soare1987recursively}.
We show that it is impossible to find computable tight upper and lower bounds on the feedback capacity of FSCs. Furthermore, we show that it is not possible to express the feedback capacity of FSCs by a finite-letter entropic expression. 

The remainder of the paper is organized as follows. The information theoretic preliminaries can be found in Section \ref{sec:preliminaries}. In Section \ref{sec:FSCs}, we introduce a class of FSCs and present the capacity results of FSCs with feedback. In Section \ref{sec:problem_formulation}, we formulate the fundamental questions that motivate our work. In Section \ref{sec:computability_numbers}, we introduce the computability framework. In Section \ref{sec:computability_analysis}, we show that the feedback capacity of FSCs is not Banach-Mazur computable. In Section \ref{sec:computability_a&c}, we show that it is not possible to find tight computable continuous upper and lower bounds, and in Section \ref{sec:finite_letter}, we show that the feedback capacity of FSCs cannot be characterized as a finite-letter maximization problem. Finally, our conclusion \tco{is given} in Section \ref{sec:conclusion}.

\subsection*{Notation}
\tco{$\N$, $\Q$, $\R$, and $\R_c$ are the sets of non-negative integers, rational numbers, real numbers, and computable real numbers; $\sP(\sX)$ and $\sP(\sY|\sX)$ denote the sets of (conditional) probability distributions on $\sY$ (given $\sX$); $H_2(\cdot)$ is the binary entropy function; $\sX^\sY$ denotes the set of all functions from $\sY$ to $\sX$.}

\section{Preliminaries}
\label{sec:preliminaries}

In this section we introduce the concept of causal conditioning and directed information, which \tco{plays} an important role in characterizing the capacity of channels with feedback. These concepts were first used and applied in \cite{massey1990causality,kramer1998directed,kramer2003capacity}.

\begin{defn}\label{eq:ccp}
The probability distribution of the sequences $x^n\in\sX^n$ \emph{causally conditioned} on the sequence $y^n\in\sY^n$ is given by
\begin{equation}
p(y^N\|x^N)=\prod_{n=1}^Np(y_n|x^n,y^{n-1}).
\end{equation}
\end{defn}

A special case of Definition \ref{eq:ccp} used in the context of the FSC with feedback is
\begin{equation}\label{eq:ccpfb}
p(y^N\|x^{N-1})=\prod_{n=1}^Np(y_n|x^{n-1},y^{n-1}).
\end{equation}

\begin{defn}
The \emph{directed information} from a sequence $X^N$ to a sequence $Y^N$ is defined by
\begin{align*}
I(X^N\rightarrow Y^N)&=\sum_{n=1}^NI(X^n;Y_n|Y^{n-1})\\
&=\sum_{n=1}^N H(X^n|Y^{n-1})-H(Y_n|X^nY^{n-1}).
\end{align*}
\end{defn}

An important property of the directed information, which we will use in our work, is that it \tco{can be} upper bounded. The  upper bound of the directed information from $X^n$ to $Y^n$ is presented in following lemma.

\begin{lem}\cite[Theorem~ 2]{massey1990causality}\label{lem:DirectedInfo}
If $X^N$ and $Y^N$ are the input and output sequences respectively of a DMC, then 
\begin{equation}
I(X^N\rightarrow Y^N) \leq \sum_{n=1}^NI(X_n;Y_n)
\end{equation}
with equality if  and only if $Y_1,Y_2,\dots,Y_n$ are statistically independent.
\end{lem}

\section{Finite State Channels with Feedback}
\label{sec:FSCs}

A suitable model to represent discrete channels with memory are discrete \tco{FSCs}. Here, we introduce the concept of \tco{FSCs} and present the capacity results with feedback known to date.

\subsection{Basic Definitions}
\label{sec:fsc_def}

Let $\sX$, $\sY$, and $\sS$ be finite input, output, and state sets. FSCs are described by following probability law
\begin{equation}
	\label{eq:probabilitylaw}
	P(y_n,s_n|x_n,s_{n-1}) \in \sP(\sY\times\sS|\sX\times\sS)
\end{equation}
where $y_n\in\sY$ and $s_n\in\sS$ are the output and state of the channel at time instant $n$ whose probability depend on the input $x_n\in\sX$ at time instant $n$ and on the previous state $s_{n-1}\in\sS$ at time instant $n-1$. We consider the transmission in presence of feedback. The feedback at time instant $n\in\N$ is the last output symbol of the channel, i.e., $y_{n-1}$; see Fig. \ref{fig:FSCfeedback}. 

For a fixed blocklength $n$,  the probability of the output sequence $y^n$ and the final state $s_n$ at time instant $n$ given an input sequence $x^n$ and an initial state $s_0$ is given by
\begin{align}
P^n(y^n,s_n|x^n,s_0)=&\sum_{s_{n-1}\in\sS}P(y_n,s_n|x_n,s_{n-1})\nonumber\\
&\quad\times P^{n-1}(y^{n-1},s_{n-1}|x^{n-1},s_0).\label{nfoldprobablity}
\end{align}

\tikzstyle{block} = [draw, rectangle, 
  minimum height=3em, minimum width=4em]

 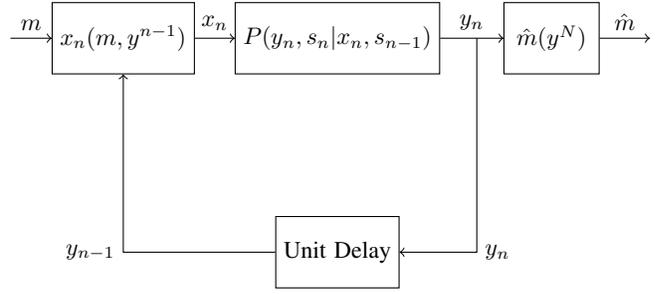
\begin{figure}
     \begin{tikzpicture}[node distance=9em, every node/.style={scale=
    0.9}]
     	
       \node [block] (a) {$x_n(m,y^{n-1})$};
       \node [block, right of=a] (b) {$P(y_n,s_n|x_n,s_{n-1})$};
       \node [block, right of=b] (c) {$\hat{m}(y^N)$};
       \node [block, below of=b] (d) {Unit Delay};

	 \draw[->,draw=black] (-1.5,0) -- node[anchor=south]{$m$} (a);    
       \draw[->,draw=black] (a) -- node[anchor=south]{$x_n$} (b);
       \draw[->,draw=black] (b) -- node[anchor=south]{$y_n$}  (c);
      \draw[->,draw=black] (c) -- node[anchor=south]{$\hat{m}$} (7,0); 
       \draw[->, draw=black]
   		(4.7,0) |- node[anchor=west]{$y_n$} (d);
   		
   	\draw[->, draw=black]
   		(d) -| node[anchor=east]{$y_{n-1}$} (a);      
     \end{tikzpicture}
      \caption{Finite state channel with deterministic feedback $y_{n-1}$.}
	\label{fig:FSCfeedback}
  \end{figure}
  
In this work we focus on unifilar FSCs.

\begin{defn}
An FSC is \tco{called \emph{unifilar}} if there exists a time-invariant function $f(\cdot)$ such that the state evolves according to the equation
\begin{equation*}
	s_n= f(s_{n-1}, x_n,y_n).
\end{equation*}
\end{defn}
\begin{rem}
The probability law of a unifilar FSC is described by 
\begin{align}
P(y_n,s_n|x_n,s_{n-1})&=W(y_n|x_n,s_{n-1})p(s_n|y_n,x_n,s_{n-1})\nonumber\\
&=W(y_n|x_n,s_{n-1})\nonumber\\
&\quad\quad\times\mathbbm{1}(s_n= f(s_{n-1}, x_n,y_n)).\label{unifilar}
\end{align}
From \eqref{unifilar}, we see that we only need the channel \tco{$W\in\sP(\sY|\sX\times\sS)$} and the transition state function $f$ to fully describe a unifilar FSC.
\end{rem}

The capacity of general FSCs with deterministic feedback was derived in \cite{permuter2009finite}. In this work, we study the algorithmic behavior of the capacity depending on the parameters $\{W,f,s_0\}$. To this aim, we express  the capacity of unifilar FSCs with feedback as a function of $\{W,f,s_0\}$, i.e., $C_{FB}(\{W,f,s_0\})$.

To describe the feedback capacity function, we introduce the upper and lower capacity as follows:
\begin{align*}
	\underbar{C}_{FB}&(\{W,f,s_0\}) \\
	&= \lim_{N\rightarrow\infty}\frac{1}{N}\max_{p(x^N\|y^{N-1})}\min_{s_0}I(X^N\rightarrow Y^N|s_0),\\
	\overline{C}_{FB}&(\{W,f,s_0\}) \\
	&= \lim_{N\rightarrow\infty}\frac{1}{N}\max_{p(x^N\|y^{N-1})}\max_{s_0}I(X^N\rightarrow Y^N|s_0).
\end{align*}

\begin{thm}\cite{permuter2009finite}\label{thm:fbcapacity}
For any unifilar FSC with deterministic feedback, the capacity is shown to be bounded by
\begin{align*}
\underbar{C}_{FB}(\{W,f,s_0\})\leq C_{FB}(\{W,f,s_0\}) \leq\overline{C}_{FB}(\{W,f,s_0\}).
\end{align*}
\end{thm}

Indecomposable FSCs are channels for which the initial state effect on the capacity  vanishes with time. To define indecomposable FSCs we set 
\begin{equation*}
q^n(s_n|x^n,s_0)=\sum_{y^n\in\mathcal{Y}^n}P^n(y^n,s_n|x^n,s_0).
\end{equation*}•
\begin{defn}
An FSC is called \emph{indecomposable} if for every $\epsilon>0$ there exists an $n_0\in\mathbb{N}$ such that for all $n\geq n_0$ we have $|q^n(s_n|x^n,s_0)-q^n(s_n|x^n,s'_0)|\leq \epsilon$ for all $s_n\in\sS$, $x^n\in\sX^n$,  $s_0\in\sS$, and  $s'_0\in\sS$.
\end{defn}		

 We next introduce the definition of strongly connected FSCs, which we consider throughout this paper.

\begin{defn}
A finite state channel is said to be \emph{strongly connected} if for any state $s$ there exists an integer $T$ and an input distribution of the form $\{p(x_n|s_{n-1})\}_{n=1}^T$ such that the probability that the channel reaches state $s$ for, any starting state $s'$, in less than $T$ time-steps, is positive, i.e.
\begin{equation*}
\tco{\sum_{n=1}^T \Pr(S_n=s|S_0=s')>0,\quad \forall s\in\mathcal{S}, \forall s' \in\mathcal{S}.}
\end{equation*}
\end{defn}

\begin{rem}
Strongly connected FSCs are also indecomposable FSCs. However, not every indecomposable FSC is also a strongly connected FSC.
\end{rem}

If a unifilar FSC is also strongly connected, and therefore indecomposable, then the lower and upper capacity coincide and are equal to the capacity, i.e., $\overline{C}_{FB}(\{W,f,s_0\})=\underbar{C}_{FB}(\{W,f,s_0\})= C_{FB}(\{W,f,s_0\})$. The capacity of indecomposable unifilar FSCs with feedback is presented in the following theorem.

\begin{thm}\label{thm:ifsccapacity}\cite{permuter2009finite}
The capacity of a indecomposable unifial FSC  with deterministic feedback is
\begin{equation}\label{eq:ifscapacity}
C_{FB}(\{W,f,s_0\})=\lim_{N\rightarrow\infty}\max_{p(x^N\|y^{N-1})}\frac{1}{N} I(X^N\rightarrow Y^N).
\end{equation}
\end{thm}

From \eqref{eq:ifscapacity} we see that the capacity of indecomposable unifilar FSCs with feedback does not have a \tco{finite-letter} representation. It is the sequence of optimization problems. For arbitrary but fixed $N\in\N$, we have that $\max_{p(x^N\|y^{N-1})}\frac{1}{N} I(X^N\rightarrow Y^N)$ is a computable number. This way, $\{\max_{p(x^N\|y^{N-1})}\frac{1}{N} I(X^N\rightarrow Y^N)\}_{N\in\N}$ is a computable sequence of computable numbers. From Theorem \ref{thm:ifsccapacity} we have that  $\{\max_{p(x^N\|y^{N-1})}\frac{1}{N} I(X^N\rightarrow Y^N)\}_{N\in\N}$ is a convergent sequence, however we do not know if it converges effectively and if $C_{FB}(\{W,f,s_0\})$ is a computable number at all.
\section{Problem Formulation}
\label{sec:problem_formulation}
The capacity function of communication scenarios has an entropic formulation. For example, for finite input and output alphabets $|\sX|<\infty$ and $|\sY|<\infty$, the capacity of a DMC $W\in\sP(\sY|\sX)$ is $\max_{p\in\sP(\sX)}I(p,W)$, i.e., the maximization of an entropic function over the input probabilities, see \cite{shannon1949communication}. Let us consider BSCs with rational crossover probability $\epsilon\in[0,1)\cap\Q$. Such a channel is clearly computable, since every rational number can be exactly expressed by a digital computer. Interestingly, the capacity of a BSC with rational crossover probability, except for \tco{$\epsilon=\{0,\frac{1}{2}\}$}, is a transcendental number.

\tcr{Subsets} of the transcendental numbers are non-computable. Intuitively, a number $x\in\R$ is computable if there exists an algorithm for $x$ that, given a desired precision, returns an approximation of the number to that precision in finitely many steps. 
 A function $g\colon\R\rightarrow\R$ is Turing computable, if there exists an algorithm that returns a computable number for every possible computable input parameter. The Blahut-Arimoto algorithm is an algorithm that takes any computable DMC as input an computes the capacity, see \cite{arimoto1972algorithm,blahut1972computation}. Hence, the capacity of DMCs is a computable function.

Coming back to the FSC with feedback, as stated in Theorem \ref{thm:fbcapacity}, the capacity of general FSCs with feedback is bounded from above and from below. Both bounds are given by a multi-letter expression. 
If we restrict the class of FSCs to indecomposable and unifilar, a mathematical expression of the capacity known. 

In Theorem \ref{thm:ifsccapacity}, the capacity of indecomposable unifilar FSCs with feedback is given by a multi-letter expression. This expression is the limit of a sequence of optimization problems. At first glance, this expression looks complicated to compute. It would be desirable to have a universal algorithm for indecomposable and unifilar FSCs that takes the channel $W\in\sP(\sY|\sX\times\sS)$, the state transition function $f\in\sS^{\sS\times\sX\times\sY}$, and the initial state $s_0\in\sS$ and computes the capacity in presence of feedback.
This is visualized in Fig. \ref{fig:TuringMachine}. We ask the following question:

{\bf Question 1:} \emph{For fixed and finite alphabets $\sX$, $\sY$, and $\sS$, is there an algorithm that takes a channel $W$, a state transition function $f$ and an initial state $s_0$ as inputs and computes the feedback capacity function $C_{FB}(\{W,f,s_0\})$?}

  \begin{figure}
     \begin{tikzpicture}       [block/.style={draw,minimum width=#1,minimum height=4em},
        block/.default=10em,high/.style={minimum height=3em},auto]

     	\node (b) at (-1,0.5)  {$W$};
       \node (c) at (-1,0) {$f$};
       \node (d) at (-1,-0.5) {$s_0$};
       \node [block, right=of c] (a) {$\fT_{C_{FB}}$};
       \node (e) at (6,0) {$C_{FB}(\{W,f,s_0\})$};


       \draw[->,draw=black] (b) -- (0.23,0.5);
       \draw[->,draw=black] (c) -- (a);
       \draw[->,draw=black] (d) -- (0.23,-0.5);
       \draw[->,draw=black] (a) -- (e);
            
     \end{tikzpicture}
     \caption{Turing machine $\fT_{C_{FB}}$ for capacity computation for fixed and finite alphabets $\sX$, $\sY$, and $\sS$. It takes the channel $W\in\sP(\sY|\sX\times\sS)$, the state transition function $f\in\sS^{\sS\times\sX\times\sY}$, and the initial state $s_0\in\sS$ and computes the capacity in presence of feedback $C_{FB}(\{W,f,s_0\}$. }
	\label{fig:TuringMachine}
  \end{figure}
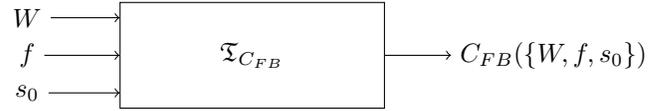

\tco{Coding schemes and general achievability results provide us with lower bounds.} Upper bounds are established via converse arguments. In \cite{shannon1967lower}, techniques to derive lower and upper bounds for the DMC were presented. This techniques are an approach to find computable tight lower and upper bounds on digital computers. 

In  practical communication scenarios, such as described in \cite{ephremides1998information},
 and standard approaches, the design, optimization and standardization of communication networks simulate the behavior of coding procedures and complex protocols and thus provide achievable performance lower bounds for optimal performance. In practice, the behavior of coding procedures is always compared with upper bounds or, if possible, optimal performances.
Such lower and upper bounds should be computable to enable a numerical evaluation on digital computers. This motivates our next question:

{\bf Question 2:}\emph{ For fixed and finite alphabets $\sX$, $\sY$, and $\sS$, is it possible to find tight computable lower and upper bounds depending on the parameters $\{W,f\}$ for the feedback capacity of unifilar FSCs?}

In \cite{shannon1949communication}, the well know capacity characterization of the DMC is formulated as a convex maximization problem over the input distribution set $\sP(\sX)$. This formulation is frequently used in communication models in information theory, e.g., the wiretap channel \cite{wyner1975wire}. There is a great interest in formulating the capacity of FSCs with feedback, among other communication scenarios, as such a convex optimization problem. This leads us to the following question: 

{\bf Question 3:} \emph{For fixed and finite alphabets $\sX$, $\sY$, and $\sS$, is it possible to characterize the capacity function of unifilar FSCs with feedback $C_{FB}(\{W,f,s_0\})$ as the optimization problem of a finite letter function?}

\section{Computability Framework}
\label{sec:computability_numbers}

To address Questions 1 and 2 and \tco{to} study the capacity of the FSC with feedback from an algorithmic perspective, we introduce the basics of computablity theory. The concepts of computability and computable real numbers were first introduced by Turing in \cite{turing1936computable} and \cite{turing1938computable}. Computable numbers are real numbers that are computable by Turing machines. 

A sequence of rational numbers $\{r_n\}_{n\in\N}$ is called a \emph{computable sequence} if there exist recursive functions $a,b,s:\N\rightarrow\N$ with $b(n)\neq0$ for all $n\in\N$ and
\begin{equation}
	\label{eq:computability_comp1}
	r_n= (-1)^{s(n)}\frac{a(n)}{b(n)}, \qquad n\in\N.
\end{equation}

A real number $x$ is said to be computable if there exists a computable sequence of rational numbers $\{r_n\}_{n\in\N}$ such that
\begin{equation}
	\label{eq:computability_comp2}
	|x-r_n|<2^{-n}
\end{equation}
for all $n\in\N$. This means that the computable real number $x$ is completely characterized by the recursive functions $a,b,s:\N\rightarrow\N$. It has the representation $(a,b,s)$ which we also write as $x\sim (a,b,s)$. It is clear that this representation must not be unique and that there might be other recursive functions $a',b',s':\N\rightarrow\N$ which characterize $x$, i.e., $x\sim (a',b',s')$.

We denote the set of computable real numbers by $\R_c$. Based on this, we define the set of computable probability distributions $\sP_c(\sX)$ as the set of all probability distributions $p\in\sP(\sX)$ such that $p(x)\in\R_c$ for every $x\in\sX$. The set of all computable conditional probability distributions $\sP_c(\sY|\sX)$ is defined accordingly, i.e., for $W:\sX\rightarrow\sP(\sY)$ we have $W(\cdot|x)\in\sP_c(\sY)$ for every $x\in\sX$. This is important since a Turing machine can only operate on computable real numbers.

We consider the capacity as a function of the tuple $\{W,f,s_0\}$. For this, we introduce the notion of computable functions.  

\begin{defn}
	\label{def:borel}
	A function $f:\R_c\rightarrow\R_c$ is called \emph{Borel-Turing computable} if there is an algorithm (or Turing machine) that transforms each given representation $(a,b,s)$ of a computable real number $x$ into a corresponding representation for the computable real number $f(x)$.
\end{defn}

To answer Question 2, we need the concept of \emph{computable continuous functions} \cite[Def.~A]{pour2017computability}. For this, let $\mathbb{I}_c$ denote a computable interval, i.e., $\mathbb{I}_c=[a,b]$ with $a,b\in\R_c$.

\begin{defn}[\cite{pour2017computability}]
	\label{def:compcont}
	Let $\mathbb{I}_c\subset\R_c$ be a computable interval. A function $f:\mathbb{I}_c\rightarrow\R_c$ is called  \emph{computable continuous} if:
	\begin{enumerate}
		\item $f$ is \emph{sequentially computable}, i.e., $f$ maps every computable sequence $\{x_n\}_{n\in\N}$ of points $x_n\in\mathbb{I}_c$ into a computable sequence $\{f(x_n)\}_{n\in\N}$ of real numbers,
		\item $f$ is \emph{effectively uniformly continuous}, i.e., there is a recursive function $d:\N\rightarrow\N$ such that for all $x,y\in\mathbb{I}_c$ and all $N\in\N$ with
		\begin{equation*}
			\|x-y\|\leq\frac{1}{d(N)}
		\end{equation*}	
		it holds that
		\begin{equation*}
			|f(x)-f(y)|\leq\frac{1}{2^N}.
		\end{equation*}
	\end{enumerate}
\end{defn}
\begin{rem}
The notion of computable continuous functions is stronger than the one of Borel-Turing computable functions. Functions that are computable continuous are also Borel-Turing computable.
\end{rem}

Computable continuous functions can be effectively approximated by sequences of computable continuous functions. This property is crucial to evaluate the answer to Question 2. The effective approximation of computable continuous functions is stated in the following corollary: 

\begin{cor}[\cite{boche2020shannon}]
	\label{cor:corollary1}
	Let $\{F_N\}_{N\in\N}$ and $\{G_N\}_{N\in\N}$ be computable sequences of computable continuous functions on $[0,1]$ with
	\begin{equation*}
		F_N(x) \leq F_{N+1}(x) \leq G_{N+1}(x) \leq G_N(x)
	\end{equation*}
	and
	\begin{equation*}
		\lim_{N\rightarrow\infty}F_N(x) = \lim_{N\rightarrow\infty}G_N(x) \eqqcolon \phi(x), \quad x\in[0,1].
	\end{equation*}
	Then $\phi:[0,1]\rightarrow\R$ is also a computable continuous function and $\{F_N\}_{N\in\N}$ and $\{G_N\}_{N\in\N}$ converge effectively to $\phi$.
\end{cor}

There are other forms of computability including \emph{Banach-Mazur computability}, which is the weakest form of computability.

\begin{defn}
	\label{def:banachmazur}
	A function $f:\R_c\rightarrow\R_c$ is called \emph{Banach-Mazur computable} if $f$ maps any given computable sequence $\{x_n\}_{n\in\N}$ of computable real numbers into a computable sequence $\{f(x_n)\}_{n\in\N}$ of computable real numbers.
\end{defn}

In particular, Borel-Turing computability  and computable continuous functions imply Banach-Mazur computability, but not vice versa.

For an overview of the logical relations between different notions of computability we refer to \cite{avigad2014computability}.

We further need the concepts of a recursive set and a recursively enumerable set as defined in \cite{soare1978recursively}. These are used with the purpose of constructing sequences of computable channels used to study the computability of the feedback capacity function.

\begin{defn}
	\label{def:recursive}
	A set $\sA\subset\N$ is called \emph{recursive} if there exists a computable function $f$ such that $f(x)=1$ if $x\in\sA$ and $f(x)=0$ if $x\notin\sA$. 
\end{defn}

\begin{defn}
	\label{def:recursiveenumerable}
	A set $\sA\subset\N$ is \emph{recursively enumerable} if there exists a recursive function whose domain is exactly $\sA$.
\end{defn}

We have the following properties \cite{soare1978recursively}:
\begin{itemize}
	\item $\sA$ is recursive is equivalent to: $\sA$ is recursively enumerable and $\sA^c$ is recursively enumerable.
	\item There exist recursively enumerable sets $\sA\subset\N$ that are not recursive, i.e., $\sA^c$ is not recursively enumerable. This means there are no computable, i.e., recursive, functions $f:\N\rightarrow\sA^c$ with $[f(\N)]=\{m\in\N\colon \exists n \in\N\;\; \text{with}\;\; f(n)=m\}=\sA^c$.
\end{itemize}
%

\section{Computability Analysis}
\label{sec:computability_analysis}
In this section, we study the existence of a Turing machine that can take any computable tuple $\{W,f,s_0\}$ and compute the feedback capacity  $C_{FB}(\{W,f,s_0\})$. In particular, we consider computable FSCs as input parameters, i.e., $P\in\sP_c$ and $W\in\sW_c$, where  $\sP_c\coloneqq \sP_c(\sY\times\sS|\sX\times\sS)$ and $\sW_c\coloneqq \sP_c(\sY|\sX\times\sS)$ denote the sets of computable conditional probabilities. We show, that the capacity of FSCs with feedback $C_{FB}(\{W,f,s_0\})$ is not even computable according to the weakest form of computability, i.e., Banach-Mazur computability. Unfortunately this result provides Question 1 with a negative answer.

\begin{thm}
For all $|\sX|\geq 2$, $|\sY|\geq 2$, and $|\sS|\geq 2$, the feedback capacity function $C_{FB}(\{W,f,s_0\})\colon\sW_c\times\sS^{\sS\times\sX\times\sY}\times\sS\rightarrow\mathbb{R}$ of unifilar FSCs with time-invariant deterministic feedback with parameters $\{W,f,s_0\}$ is not Banach-Mazur computable.
\end{thm}

The fact that the feebdack capacity of FSCs is not Banach-Mazur computable implies automatically that the feedback capacity of FSCs is not \tcr{Borel-Turing} computable. This leads us to the following \tcr{corollary}:

\begin{cor}\label{cor:notturingcomp}
    For all $|\sX|\geq 2$, $|\sY|\geq 2$, and $|\sS|\geq 2$, the feedback capacity function $C_{FB}(\{W,f,s_0\})\colon\sW_c\times\sS^{\sS\times\sX\times\sY}\times\sS\rightarrow\mathbb{R}$ of unifilar FSCs with time-invariant deterministic feedback with parameters $\{W,f,s_0\}$ is not \tcr{Borel-Turing computable.}
\end{cor}

Corollary \ref{cor:notturingcomp} states that the capacity of FSCs with feedback is not Borel-Turing computable. This implies that there is no Turing machine that for fiexed \tcr{alphabets} $|\sX|\geq 2$, $|\sY|\geq 2$, and $|\sS|\geq 2$, takes $\{W,f,s_0\}$ as inputs and computes the capacity $C_{FB}(\{W,f,s_0\})$. This gives us a negative answer to Question~1.

\begin{proof}
We consider the set of computable FSCs. The capacity is a function $C_{FB}\colon\sW_c\times\sS^{\sS\times\sX\times\sY}\times\sS\rightarrow \R_{0+}$. To prove the computability we use an indirect proof. We assume that the feedback capacity $C_{FB}$ is Borel-Turing computable and we prove the opposite by contradiction. The proof is organized as follows:

\begin{itemize}
\item  We design a suitable class of rational unifilar FSCs $\{W_\lambda,f\}_{\lambda\in [0,\frac{1}{2}]\cap\Q}$ characterized by the parameter $\lambda$. 
 
 \item We consider a recursively enumerable non recursive set $\sA$. The elements of the recursive enumerable set are listed by a unique recursive function $\varphi_A\colon \N\rightarrow\N$. There is a Turing machine $\fT_\sA$ that stops for input $n$ if and only if $n\in\sA$. Otherwise $\fT_\sA$ runs forever.
 
 \item\tco{ We generate a computable double sequences of rational numbers $\{\lambda_{n,m}\}_{{n,m}\in\N}$ using the Turing machine $\fT_{\sA}$. We use $\{\lambda_{n,m}\}_{{n,m}\in\N}$ to construct a computable double sequence of rational unifilar FSCs $\{W_{\lambda_{n,m}},f\}_{{n,m}\in\N}$} from the class of unifilar FSCs $\{W_\lambda,f\}_{\lambda\in [0,\frac{1}{2}]\cap\Q}$. This sequence of rational unifilar FSCs converges effectively to the computable sequence of computable FSCs $\{W_{\lambda_{n}^*},f\}_{{n}\in\N}$. Hence, the set $\sA$ is encoded in the sequence $\{W_{\lambda_{n}^*},f\}_{{n}\in\N}$.
 
 \item We define the function $\phi(\{W,f\})=C_{FB}(\{W,f,0\})-C_{FB}(\{W,f,1\})$. Since $C_{FB}$ is assumed to be a computable function, then $\phi$ is also Borel-Turing computable. This would mean that the sequence $\{\phi(\{W_{\lambda_{n}^*},f\})\}_{n\in\N}$  is a computable sequence of computable reals. With this computable sequence of computable reals we can build a Turing machine $\fT_*$ that stops for input $n$ if and only if $\phi(\{W_{\lambda_{n}^*},f\})>0$. Thus, $\fT_*$ stops if $n\in\sA^c$ which is a contradiction, since it would mean that $\sA$ a recursive set. Hence, the  assumption that $C_{FB}(\{W,f,s_0\})$ is computable is wrong. Even if $C_{FB}$ were Banach-Mazur computable, then it would yield a solution the halting problem, which has been proven to be unsolvable.
\end{itemize}

We first introduce the concept of distance. For two FSCs $P_1,P_2\in\sP(\sY\times\sS|\sX\times \sS)$ we define the distance between $P_1$ and $P_2$ based on the total variation distance as 
\begin{align*}
D(P_1&,P_2)\\
&=\max_{s'\in\sS}\max_{x\in\sX}\sum_{y\in\sY}\sum_{s\in\sS}|P_1(y,s|x,s')-P_2(y,s|x,s')|.
\end{align*}

We start proving the result for $|\sX|=|\sY|=|\sS|=2$. Then we extend it to $|\sX|\geq 2$, $|\sY|\geq 2$,
 and $|\sS|\geq 2$.

We consider the channel
	\begin{equation}
		W(y_n|x_n,0) = \begin{pmatrix}
			1 & 0 \\
			0 & 1
		\end{pmatrix}\!,
		W(y_n|x_n,1) = \begin{pmatrix}
			1\!-\!\epsilon & \epsilon \\
			0 & 1
		\end{pmatrix}
		\label{eq:p}
	\end{equation}
for some $\epsilon\in (0,\frac{1}{2})\cap\mathbb{Q}$, i.e., for state $s_{n-1}=0$ the channel is noiseless; for $s_{n-1}=1$ it is noisy. Further, we consider the state transition function $f\colon\sS\times\sX\times\sY\rightarrow\sS$ described by the state diagram in Fig. \ref{fig:stf}. The nodes represent the states and the tuple of the edges represent the input and output symbols $(x_n,y_n)$ of the channel.
	
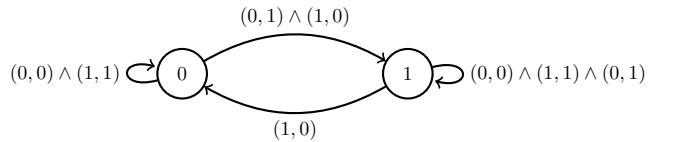
\begin{figure}[h]
\centering
	\begin{tikzpicture}[thick,scale=0.75, every node/.style={scale=0.75}]
	\node[state] [] (s0) {$0$};
	\node[state, right of=s0, xshift=3cm] (s1) {$1$};

	\draw (s0) edge [loop left] node {$(0,0)\land(1,1)$} (s0);
		\path[->] (s0) edge  [bend left, above] node {$(0,1)\land(1,0)$} (s1);
		\path[->] (s1) edge [loop right] node {$(0,0)\land(1,1)\land(0,1)$} (s1);
		 \path[->] (s1) edge [bend left, below] node {$(1,0)$} (s0);
	\end{tikzpicture}
	\caption{Diagram of the state transition function $f$.}
	\label{fig:stf}
\end{figure}

The state of the channel $s_n$ given the input $x_n$, output $y_n$ and previous state $s_{n-1}$ is also shown in Table \ref{statetranstable}.

\begin{table}
\centering
\caption{The state $s_n$ given $x_n$, $y_n$ and $s_{n-1}$.}
\label{statetranstable}
 \begin{tabular}{||c|c c c c c c c c||} 
 \hline
 $x_n$ & 0 & 0 & 0 & 0 & 1 & 1 & 1 & 1 \\  

 $y_n$ & 0 & 0 & 1 & 1 & 0 & 0 & 1 & 1 \\ 
 
 $s_{n-1}$ & 0 & 1 & 0 & 1 & 0 & 1 & 0 & 1 \\[0.5ex]
 \hline\hline
 $s_n$ & 0 & 1 & 1 & 1 & 1 & 0 & 0 & 1 \\
 \hline  
\end{tabular}

\end{table}

The channel $\{W,f,0\}$ corresponds to a discrete memoryless channel. Since the initial state is $s_0=0$, i.e. the channel $W(y_n|x_n,0)$ is noiseless,  then the only two possible input output tuples are $(0,0)$ and $(1,1)$. Applying the state transition function $f$ to the tuples $(x_1,y_1,s_0)\in\{(0,0,0), (1,1,0)\}$ we get that for both tuples the next state is $s_1=0$. This implies that if the initial state is $0$, then the state stays $0$ forever.
Applying Lemma \ref{lem:DirectedInfo} to the directed information for this channel, we get
\begin{align*}
I(X^N\rightarrow Y^N|s_0=0) = I(X^N;Y^N|s_0=0).
\end{align*}
This and the fact that $W(y|x,0)$ is a binary noiseless channel imply that the FSC feedback capacity with time-invariant deterministic feedback and initial state $s_0=0$ is
\begin{align*}
C_{FB}(\{W,f,0\}) = 1.
\end{align*}

The channel $\{W,f,1\}$ corresponds to the discrete memoryless channel $W(y|x,1)$ with $x\in\sX,y\in\sY$. \tco{Similar} to the line of arguments for $\{W,f,0\}$, if the initial state is  $s_0=1$ then the channel has only three possible input output tuples  $(0,0),(0,1)$ and $(1,1)$. Applying the state transition function $f$ to the tuples $(x_1,y_1,s_0)\in\{(0,0,1), (0,1,1), (1,1,1)\}$ we get that for all three tuples the next state is $s_1=1$. Meaning that if the initial channel state is $1$, then the channels stays in state $1$ forever. Applying Lemma \ref{lem:DirectedInfo} we have that 
\begin{equation}\label{eq:equality_dirinf}
I(X^N\rightarrow Y^N|s_0=1) = I(X^N;Y^N|s_0=1).
\end{equation}

Note that the channel $\{W,f,1\}$ is a Z-channel. Due to \eqref{eq:equality_dirinf} we see that the FSC feedback capacity for the case of time-invariant deterministic feedback and initial state $s_0=1$  is 
\begin{align*}
C_{FB}(\{W,f,1\}) &=\max_{p\in\sP(\sX)}H_2(p(1-\epsilon))-pH_2(\epsilon)\\
&=\log_2\Big(1+2^{-g(\epsilon)}\Big)
\end{align*}
with $g(\epsilon)=\frac{H_2(\epsilon)}{1-\epsilon}$. The optimal input distribution is $p(0)=\Big[(1-\epsilon)\Big(1+2^{\frac{H_2(\epsilon)}{(1-\epsilon)}}\Big)\Big]^{-1}$ and $p(1)=1-\Big[(1-\epsilon)\Big(1+2^{\frac{H_2(\epsilon)}{(1-\epsilon)}}\Big)\Big]^{-1}$.

Next we show that $C_{FB}(\{W,f,0\})$ and $C_{FB}(\{W,f,1\})$ cannot be simultaneously Banach-Mazur computable.
Let $\sA\in\N$ be an arbitrary recursively enumerable but not recursive set. Let  $\fT_\sA$ be a Turing machine that stops for input $n$ if and only if $n\in\sA$.  Otherwise $\fT_\sA$ runs forever. Such a Turing machine can easily be found as argued next: Let $\varphi_\sA:\N\rightarrow\N$ be a recursive function that lists all elements of the set $\sA$ and for which $\varphi_\sA:\N\rightarrow\sA$ is a unique function.
	
	Let $n\in\N$ be arbitrary. The Turing machine $\fT_\sA$ with input $n$ is defined as follows: We start with $l=1$ and compute $\varphi_\sA(1)$. If $n=\varphi_{\sA}(1)$, then the Turing machine stops. In the other case, the Turing machine computes $\varphi_\sA(2)$. Similarly, if $n=\varphi_{\sA}(2)$, then the Turing machine stops and otherwise, it continues computing the next element. It is clear that this Turing machine stops if and only if $n\in\sA$. 

 Assume $C_{FB}(\{W,f,0\})$ and $C_{FB}(\{W,f,1\})$ are both Borel-Turing computable. For $\lambda\in (0,\frac{1}{2}]\cap\R_c$ we consider the channel $W_\lambda\in\sW_c$ with

 \begin{align}
		W_\lambda(y_n|x_n,0) &= \begin{pmatrix}
			1-\lambda & \lambda \\
			\lambda & 1-\lambda
		\end{pmatrix},\nonumber\\
		W_\lambda(y_n|x_n,1) &= \begin{pmatrix}
			1-\epsilon & \epsilon \\
			\lambda & 1-\lambda
		\end{pmatrix}. \label{eq:Wlambda}
	\end{align}
For  $\lambda\in [0,\frac{1}{2}]\cap\tcr{\Q}$ both $W_\lambda(y_n|x_n,0)$ and $W_\lambda(y_n|x_n,1)$ are by all means computable probability distributions.

For $\lambda = 0$, we have 
\begin{align*}
	C_{FB}(\{W_0,f,0\})-&C_{FB}(\{W_0,f,1\}) 
	\\&= 1 -\log_2\Big(1+2^{-g(\epsilon)}\Big) > 0.
\end{align*}

Note that for $\lambda\in (0,\frac{1}{2}]\cap\tcr{\Q}$ and $f$ as described in Fig. \ref{fig:stf} and Table \ref{statetranstable} the FSC ${\{W_\lambda,f,s_0\}}$ is strongly connected. For $s_{0}=0$  the FSC achieves the state $s_1=1$ if the input output tuple $(x_1,y_1)$ of the channel is  either $(0,1)$ or $(1,0)$. Since $W_\lambda(1|0,0)=W_\lambda(0|1,0)=\lambda>0$ the channel can reach the state $s_1=1$ in one time-step for any input distribution $p(x_1|0)>0$. Similarly, for  $s_{0}=1$,  the FSC achieves the state $s_1=0$ if the input output tuple $(x_1,y_1)$ of the channel is $(1,0)$. Since $W_\lambda(0|1,1)=\lambda>0$, the channel can reach the state $s_1=0$ in one time-step for any input distribution with $p(0|0)>0$.
This implies that the FSC ${\{W_\lambda,f,s_0\}}$ with $s_0\in\sS$ is also indecomposable. 

Hence, for every $\lambda\in(0,\frac{1}{2}]$ we have
\begin{equation}
	C_{FB}(\{W_\lambda,f,0 \})=C_{FB}(\{W_\lambda,f,1\}).
\end{equation}

Next, we generate the indirect proof. First, we build a Turing machine, that encodes the recursively enumerable non recursive set $\sA$ in a sequence of unifilar FSCs. For the construction of the Turing machine, we rely on the  construction introduced in \cite{pour1960comparison}[~Case I, page 336]. This plays an important role in emphasizing the properties of the capacity of FSCs with feedback. Similar constructions have been developed in \cite{boche2020denial} and \cite{boche2019algorithmic}.

Let $n\in\N$ be arbitrary, for every $n\in\N$ and $m\in\N$ let
			\begin{align*}
				\lambda_{n,m} = \begin{cases}
				\frac{1}{2^l} & \fT_\sA \text{ stops for input } n \text{ after } l\leq m \text{ steps} \\
				\frac{1}{2^m} &\fT_\sA \text{ does not stop for input } n \text{ after } m \text{ steps}. 
				\end{cases}
			\end{align*}
Then the sequence \tco{$\{\lambda_{n,m}\}_{n,m\in\N}$} is a computable double sequence of rationals. For arbitrary $n\in\N$ and for all $m\geq M$, $m,M\in\N$, we have
\begin{equation}\label{effconv}
|\lambda _{n,m}-\lambda_{n,M}|<\frac{1}{2^M}.
\end{equation}•
To prove \eqref{effconv} we will consider both cases: $\fT_\sA$ stops for input $n$ and $\fT_\sA$ does not stop for input $n$.
\begin{itemize}
\item \emph{$\fT_\sA$ stops for input $n$ after $l\leq M$ iterations}: In this case $\lambda_{n,M}=\lambda_{n,m}$ so $|\lambda _{n,m}-\lambda_{n,M}|=0$.
\item  \emph{$\fT_\sA$ has not stopped for input $n$ after $M$ iterations}: For every $m\geq M$ it holds that $\lambda_{n,M}\geq\lambda_{n,m}$, meaning that $|\lambda _{n,m}-\lambda_{n,M}|=\lambda_{n,M}-\lambda_{n,m}=\frac{1}{2^M}-\lambda_{n,m}<\frac{1}{2^M}$.
\end{itemize}•
The sequence \tco{$\{\lambda_{n,m}\}_{n,m\in\N}$} is a computable double sequence of rationals that converges effectively in $m$. This implies that for every $n\in\N$ the sequence  $\{\lambda_{n,m}\}_{m\in\N}$ converges effectively to its limit $\lambda^*_n$ and the limit is a computable real number $\lambda^*_n\in\R_c$. Since \tco{$\{\lambda_{n,m}\}_{n,m\in\N}$} is a computable double sequence of rationals such that as $m\rightarrow\infty$, $\lambda_{n,m}\rightarrow\lambda^*_n$, then $\{\lambda_n^*\}_{n\in\N}$ is a sequence of computable real numbers.  It further holds $\lambda_n^*\geq0$ with equality if and only if the Turing machine $\fT_\sA$ does not stop for input $n$. 

We consider the computable double sequence of rational unifilar FSCs $\{P_{\lambda_{n,m}}\}_{n,m\in\N^2}$ defined by the computable double sequence of rational channels \tco{$\{W_{\lambda_{n,m}}\}_{n,m\in\N}$} and the function $f$ defined in Table \ref{statetranstable}. 

For arbitrary $n\in\N$ and for all $m\geq M$, $m,M\in\N$, we have
\begin{align}
D&(P_{\lambda _{n,m}},P_{\lambda _{n,M}}) \nonumber\\
&= \max_{s'\in\sS}\max_{x\in\sX}\sum_{y\in\sY}\sum_{s\in\sS}|P_{\lambda _{n,m}}(y,s|x,s')-P_{\lambda _{n,M}}(y,s|x,s')|\nonumber\\
&= \max_{s'\in\sS}\max_{x\in\sX}\sum_{y\in\sY}\sum_{s\in\sS}|W_{\lambda _{n,m}}(y|x,s')-W_{\lambda _{n,M}}(y|x,s')|\nonumber\\
&\quad\quad\times\mathbbm{1}(s= f(s', x,y))\label{eq_dif1}\\
&= \max_{s'\in\sS}\max_{x\in\sX}\sum_{y\in\sY}|W_{\lambda _{n,m}}(y|x,s')-W_{\lambda _{n,M}}(y|x,s')|\nonumber\\
& = 2|\lambda _{n,m}-\lambda _{n,M}|<\frac{1}{2^{M-1}}\label{eq_dif2}
\end{align}
where \eqref{eq_dif1} holds since \tco{$\{P_{\lambda_{n,m}}\}_{n,m\in\N}$} are unifiar. \eqref{eq_dif2} results from \eqref{eq:Wlambda} in the following way: 
\begin{align*}
\sum_{y\in\sY}|W_{\lambda _{n,m}}&(y|0,0)-W_{\lambda _{n,M}}(y|0,0)|\\
&=\sum_{y\in\sY}|W_{\lambda _{n,m}}(y|1,0)-W_{\lambda _{n,M}}(y|1,0)|\\
&=\sum_{y\in\sY}|W_{\lambda _{n,m}}(y|1,1)-W_{\lambda _{n,M}}(y|1,1)|\\
&=|(1-\lambda _{n,m})-(1-\lambda _{n,M})|+|\lambda _{n,m}-\lambda _{n,M}|\\
&=2|\lambda _{n,m}-\lambda _{n,M}|
\end{align*}
and 
\begin{equation*}
\sum_{y\in\sY}|W_{\lambda _{n,m}}(y|0,1)-W_{\lambda _{n,M}}(y|0,1)|=0.
\end{equation*}

As a result of \eqref{effconv} and \eqref{eq_dif2}, we have that ${W_{\lambda_{n,m}}}\rightarrow{W_{\lambda^*_n}}$ as $m\rightarrow\infty$,  for every $n\in\N$. Hence $\{W_{\lambda_n^*}\}_{n\in\N}$ is a sequence of computable channels. 

Further we use a sequence of computable unifilar FSCs $\{P_{\lambda^*_n}\}_{n\in\N}=\{W_{\lambda^*_n},f\}_{n\in\N}$ where the transition state function is fixed.

Since $C_{FB}(\{W,f,0\})$ and $C_{FB}(\{W,f,1\})$ are assumed to be Borel-Turing computable functions, the difference \tco{$\phi(\{W,f\})=C_{FB}(\{W,f,1\})-C_{FB}(\{W,f,0\})$} is a Borel-Turing computable function as well. Then, the sequence $\{\mu_n\}_{n\in\N}$ with 
	\begin{equation*}
		\mu_n = \tco{\phi(\{W_{\lambda_n^*},f\})}, \quad n\in\N,
	\end{equation*}
	is a computable sequence of computable real numbers. With this, we find a computable double sequence $\{\nu_{n,m}\}_{n,m\in\N}$ of rational numbers with
	\begin{equation*}
		\big|\mu_n-\nu_{n,m}\big| < \frac{1}{2^m}.
	\end{equation*}
	For every $n$, we can consider the following Turing machine $\fT_*$: For input $n$, we set $m=1$ and check if
	\begin{equation*}
		\nu_{n,1} > \frac{1}{2}
	\end{equation*}
	is satisfied. If this is true, the Turing machine stops. Otherwise, we set $m=2$ and check if
	\begin{equation*}
	\nu_{n,2} > \frac{1}{4}
	\end{equation*}
	is satisfied. If this is true, the Turing machine stops. Otherwise, it continues as described. Next, we show that this Turing machine $\fT_*$ stops for input $n$ if and only if $\mu_n>0$.
	
	``$\Leftarrow$'' If $\mu_n>0$, then there exists an $\tilde{M}$ with
	\begin{equation*}
		\frac{1}{2^{\tilde{M}}} < \frac{\mu_n}{2}
	\end{equation*}
	so that
	\begin{align*}
		\mu_n &= \mu_n - \nu_{n,\tilde{M}} + \nu_{n,\tilde{M}} \leq \big|\mu_n - \nu_{n,\tilde{M}}\big| + \nu_{n,\tilde{M}} \\
			&< \frac{1}{2^{\tilde{M}}} + \nu_{n,\tilde{M}} < \frac{\mu_n}{2} + \nu_{n,\tilde{M}}, 
	\end{align*}
	i.e., the Turing machine $\fT_*$ stops for input $n$ within $\tilde{M}$ steps.
	
	``$\Rightarrow$'' It holds $\nu_{n,\hat{M}}>\frac{1}{2^{\hat{M}}}$ for a certain $\hat{M}$. Then,
	\begin{align*}
		\frac{1}{2^{\hat{M}}} &< \nu_{n,\hat{M}} = \nu_{n,\hat{M}} - \mu_n + \mu_n \\
			&\leq \big|\nu_{n,\hat{M}} - \mu_n\big| + \mu_n < \frac{1}{2^{\hat{M}}} + \mu_n
	\end{align*}
	so that $\mu_n>0$ is true. 

This means that there exists a Turing machine $\fT_\sS$ with
		\begin{align*}
		\fT_\sS(n) = \begin{cases}
			\tco{n\in\sA} &\text{ if } \fT_\sA \text{ stops for input } n \\
			\tco{n\in\sA^c} &\text{ if } \fT_* \text{ stops for input } n.
			\end{cases}
	\end{align*}
	This implies that $\sA$ is a recursive set, which is a contradiction. This contradiction shows that $C_{FB}(\{W,f,0\})$ and $C_{FB}(\{W,f,1\})$ cannot be Banach-Mazur computable. This immediately implies that they cannot be Borel-Turing computable as well.
	
	To extend the proof to $|\sX|\geq 2$, $|\sY|\geq 2$, and $\sS\geq 2$, we will divide the extension in two steps: 
	\begin{itemize}
		\item The state set remains binary and the input and output alphabets may grow, i.e.,  $|\sS|= 2$ and $|\sX|\geq 2$, $|\sY|\geq 2$.
		\item We allow the state set to grow, i.e., $|\sS|\geq 2$.
	\end{itemize}
	
	\textbf{Step I}: For $|\sX|\geq 2$, $|\sY|\geq 2$, and $|\sS|=2$ arbitrary, we take the sequence of parameters $\{W_\lambda, f\}$ as above and extend them as follows: We set $W_\lambda(y_n|x_n,s_{n-1}) =0 $ for $y_n\in\sY\setminus\{0,1\}$, $x_n\in\sX$ and $s_{n-1}\in\sS$ and also for $y_n\in\sY$, $x_n\in\sX\setminus\{0,1\}$ and $s_{n-1}\in\sS$. For every pair $(x_n,y_n)\in(\sX\setminus\{0,1\}\times\sY )\cup(\sX\times\sY\setminus\{0,1\} )$ we define the transition state function to be $f(x_n,y_n,s_{n-1})=s_{n-1}$.
	
	\textbf{Step II}: Let $|\sS|\geq 2$. For every $s\in\sS\setminus\{0,1\}$ set 
	\begin{align*}
		W_\lambda(0,|0,s)&=1-\Big(\epsilon+\Big(\frac{1}{2}-\epsilon\Big)^{s-1}\Big),\\
		 W_\lambda(1,|0,s)&=\epsilon+\Big(\frac{1}{2}-\epsilon\Big)^{s-1},\\
		W_\lambda(0,|1,s)&=0,\quad W_\lambda(1,|1,s)=1.
	\end{align*}
	Note that for $(x,y)\in\{0,1\}^2$, for every $s$ we have the Z-channel with probability of transmitting bit $0$ incorrectly of $\delta_s= \epsilon+\Big(\frac{1}{2}-\epsilon\Big)^{s-1}$. It also holds, that for every $s\in\sS\setminus\{0,1\}$, $\epsilon<\delta_s\leq \frac{1}{2}$, hence the channel at state $s=1$ is less noisy than the channels at states $s\geq 2$.
	
	We set $W_\lambda(y_n|x_n,s) =0 $ for $y_n\in\sY\setminus\{0,1\}$, $x_n\in\sX$ and $s_{n-1}\in\sS$ and also for $y_n\in\sY$, $x_n\in\sX\setminus\{0,1\}$ and $s_{n-1}\in\sS$.
	
	Next we modify the state transition function $f$ as follows:
For $|\sS|=3$, we have a new state $s=2$. For $(x_n,y_n,s_{n-1})=(0,1,0)$, we modify the function $f$ by setting the next state $s_{n}$ to be $s_{n}=f(0,1,0)=2$. For the state $s_{n-1}=2$ we complete the state transition function
 \begin{equation*}
 f(x_n,y_n,2)= \begin{cases}
        0\quad \text{for every $(x_n,y_n)$ s.t. $W(x_n,y_n,2)>0$,}
        \\
        2 \quad \text{for every $(x_n,y_n)$ s.t. $W(x_n,y_n,2)=0$}.
        \end{cases}
 \end{equation*}
The diagram of the state transition function $f$ for  $|\sX|=|\sY|=2$ and $|\sS|=3$ is illustrated in Fig. \ref{fig:stf2}.
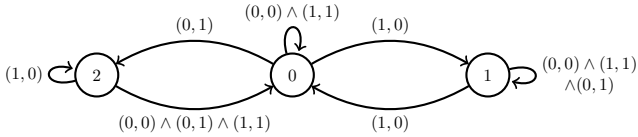
\begin{figure}[h]
\centering
	\begin{tikzpicture}[thick,scale=0.7, every node/.style={scale=0.65}]
	\node[state] [] (s0) {$0$};
	\node[state, right of=s0, xshift=3cm] (s1) {$1$};
	\node[state, left of=s0, xshift=-3cm] (s2) {$2$};

		\draw (s0) edge [loop above] node {$(0,0)\land(1,1)$} (s0);
		\path[->] (s0) edge  [bend left, above] node {$(1,0)$} (s1);
		\path[->] (s0) edge  [bend right, above] node {$(0,1)$} (s2);
		\path[->] (s1) edge [loop right]  node[align=center] {$(0,0)\land(1,1)$\\$\land(0,1)$} (s1);
		 \path[->] (s1) edge [bend left, below] node {$(1,0)$} (s0);
		 \path[->] (s2) edge  [bend right, below] node {$(0,0)\land(0,1)\land(1,1)$} (s0);
		 \draw (s2) edge [loop left] node {$(1,0)$} (s2);
	\end{tikzpicture}
	\caption{Diagram of the state transition function $f$ for $|\sX|=|\sY|=2$ and $|\sS|=3$.}
	\label{fig:stf2}
\end{figure}

If $|\sS|\geq 4$, we extend the transition function iteratively as described above:
\begin{itemize}
\item Let $2 \leq s<|\sS|-1$. For $s_{n-1}=s$ we set $s_n=f(x_n,y_n,s_{n-1})$  to be 
	 \begin{equation*}\resizebox{.96\hsize}{!}{$
 \!\!\!\! f(x_n,y_n,s)\!=\! \begin{cases}
        \!0 \;\;\;\text{$(x_n,y_n)\in(\sX\times\sY)\setminus \{(0,1)\}$ s.t. $W(x_n,y_n,s)>0$,}
        \\
        \! s \;\;\; \text{$(x_n,y_n)\in(\sX\times\sY)\tcr{\setminus \{(0,1)\}}$  s.t. $W(x_n,y_n,s)=0$,}\\
        \! s+1 \;\;\;\text{$(x_n,y_n)=(0,1)$  s.t. $W(x_n,y_n,s)=0$.}
        \end{cases}$}
 \end{equation*}
 \item For $s =|\sS|-1$ and $ s \geq 2$ we have 
  \begin{equation*}\resizebox{.96\hsize}{!}{$
 f(x_n,y_n,s)= \begin{cases}
        0\;\;\; \text{$(x_n,y_n)\in\sX\times\sY$ s.t. $W(x_n,y_n,s)>0$,}
        \\
        s \;\;\; \text{$(x_n,y_n)\in\sX\times\sY$  s.t. $W(x_n,y_n,s)=0$}.
        \end{cases}$}
 \end{equation*}
\end{itemize}
This way the FSCs with $|\sX|\geq 2,$ $|\sY|\geq 2$, and $|\sS|\geq 2$ preserve the properties of the FSCs constructed above, i.e., they are unifilar and strongly connected.
	\end{proof}

\begin{rem}
We showed that capacity of FSCs with feedback is not Banach-Mazur computable for a special class of FSCs, the unifilar FSCs. This result holds for more general classes of FSCs as well.
\end{rem}

\section{Computability Analysis of Achievability and Converse}
\label{sec:computability_a&c}
In \tco{the previous section}, we showed that the capacity of FSCs with feedback is not a computable function. \tco{Here}, we are interested in finding computable tight upper and lower bounds on the feedback capacity function and therefore a computable representation of achievability and converse. We show that it is not possible to find upper and lower bounds that are simultaneously computable. This provides us with a negative answer to Question 2.

\begin{thm}
	\label{the:fsctheorem5}
	For $|\sX|\geq2$, $|\sY|\geq2$, and $|\sS|\geq2$ arbitrary but fixed,  there exists an $s_0\in\sS$ such that the following holds: There exists no computable sequences $\{F_N\}_{N\in\N}$ and $\{G_N\}_{N\in\N}$ of computable continuous functions with
	\begin{enumerate}
		\item $F_N:\tco{\sW_c}\times\sS^{\sS\times\sX\times\sY}\rightarrow\R$ and $G_N:\tco{\sW_c}\times\sS^{\sS\times\sX\times\sY}\rightarrow\R$, $N\in\N$,
		\item $F_N(W,f)\leq C_{FB}(\{W,f,s_0\})$, $W\in\tco{\sW_c}$, $f\in\sS^{\sS\times\sX\times\sY}$, $N\in\N$, and $\lim_{N\rightarrow\infty}F_N(W,f)=C_{FB}(\{W,f,s_0\})$ for all $W\in\tco{\sW_c}$, $f\in\sS^{\sS\times\sX\times\sY}$,
		\item $C_{FB}(\{W,f,s_0\})\leq G_N(W,f)$, $W\in\tco{\sW_c}$, $f\in\sS^{\sS\times\sX\times\sY}$, $N\in\N$, and $\lim_{N\rightarrow\infty}G_N(W,f)=C_{FB}(\{W,f,s_0\})$ for all $W\in\tco{\sW_c}$, $f\in\sS^{\sS\times\sX\times\sY}$.
	\end{enumerate}
\end{thm}
\begin{proof}
	The result follows immediately from Corollary \ref{cor:corollary1}. If such sequences $\{F_N\}_{N\in\N}$ and $\{G_N\}_{N\in\N}$ would exist, then $C_{FB}$ would be a computable continuous function which is a contradiction, since $C_{FB}$ is for a certain $s_0\in\sS$ not Banach-Mazur computable. 
\end{proof}

This result shows that an approximation of $C_{FB}$ by computable continuous functions is not possible. From this, we can immediately conclude the following.

\begin{cor}
	\label{cor:fsccor2}
	For all computable sequences $\{F_N\}_{N\in\N}$ and $\{G_N\}_{N\in\N}$ of computable continuous functions for which there exists an $s_0\in\sS$ such that for $N\in\N$ it holds that
	\begin{equation*}
		F_N(W,f) \leq C_{FB}(\{W,f,s_0\})
	\end{equation*}
	for all $W\in\tco{\sW_c}$ and $f\in\sS^{\sS\times\sX\times\sY}$, and for $N\in\N$ it holds that
	\begin{equation*}
		C_{FB}(\{W,f,s_0\}) \leq G_N(W,f)
	\end{equation*}
	for all $W\in\tco{\sW_c}$ and $f\in\sS^{\sS\times\sX\times\sY}$, there must exist a $(W_*,f_*)\in\tco{\sW_c}\times\sS^{\sS\times\sX\times\sY}$ such that
	\begin{equation}
	\begin{split}
		0 < \max\Big\{&\limsup_{N\rightarrow\infty}\big|C_{FB}(\{W_*,f_*,s_0\})-F_N(W_*,f_*)\big|,\\
			&\quad \limsup_{N\rightarrow\infty}\big|C_{FB}(\{W_*,f_*,s_0\})-G_N(W_*,f_*)\big|\Big\}.
		\label{eq:fsccor}
	\end{split}
	\end{equation}
\end{cor}
\begin{proof}
	These statements follow immediately from Theorem \ref{the:fsctheorem5}, since if \eqref{eq:fsccor} would be zero for all $(W,f)\in\tco{\sW_c}\times\sS^{\sS\times\sX\times\sY}$, then this would imply that $C_{FB}$ is a computable function.
\end{proof}

The functions $\{F_N\}$ can be interpreted as lower bounds for achievable rates and the capacity, while $\{G_N\}$ can be interpreted as upper bounds for the achievable rates and the capacity. Contrary to DMCs, Corollary \ref{cor:fsccor2} states that it is impossible to find a Turing machine that takes any $\{W,f\}$ as input and computes tight upper and lower bounds for capacity of FSCs with feedback. There is either no computable achievability or no computable converse (or both are not computable).

One cannot find techniques, such as the ones for DMCs, that can be implemented on digital computer and gives us, up to a certain precision, the range in which the optimal performance lies. Consequently, if one is interested in studying the behavior of a coding procedure for FSCs with feedback, it is impossible to numerically evaluate it by comparing it to tight bounds of its optimal performance.


\section{Feedback Capacity as a Finite Multi-Letter Optimization Problem}
\label{sec:finite_letter}
In this section, we study whether or not it is possible to formulate the capacity of FSCs with feedback as a finite multi-letter optimization problem. To this aim, we first have to study the continuity behavior of the capacity function. We show that the capacity function is discontinuous for certain $s_0\in\sS$, $f\in\sS^{\sS\times\sX\times\sY}$ and computable $W\in\sW_c$. The discontinuity result makes it impossible to describe the capacity of FSCs with feedback as a finite multi-letter optimization problem providing us with a negative answer to Question 3.
\begin{thm}
	\label{thm:discont}
	For all $|\sX|\geq2$, $|\sY|\geq2$, and $|\sS|\geq2$, the capacity function $C_{FB}:\tco{\sP(\sY|\sX\times\sS)}\times\sS^{\sS\times\sX\times\sY}\rightarrow\R$ is discontinuous.
\end{thm}

\begin{proof}
	We consider the channels $W(y_n|x_n,0)$ and $W(y_n|x_n,1)$ as in \eqref{eq:p} and the state transition function $f$ as described in Table \ref{statetranstable}.
	
	Next, we consider $\{W_k,f,s_0\}$ for $k\geq1$ with 
 \begin{align}\label{eq:Wk}
		W_k(y_n|x_n,0) &= \begin{pmatrix}
			1-\frac{1}{k} & \frac{1}{k} \\
			\frac{1}{k} & 1-\frac{1}{k}
		\end{pmatrix}\\
		W_k(y_n|x_n,1) &= \begin{pmatrix}
			1-\epsilon & \epsilon \\
			\frac{1}{k} & 1-\frac{1}{k}
		\end{pmatrix}.
	\end{align}
	We observe that the FSC $\{W_k,f,s_0\}$, $s_0\in\sS$, $k\geq1$, as defined above is unifilar and strongly connected, and therefore indecomposable.  Note that for every $x_n\in\sX,y_n\in\sY$ and $s_0\in\sS$ we have $W_k(y_n|x_n,s_0)\in\sW_c$, which implies that the channels are computable.

For FSCs as defined in \eqref{eq:p}-\eqref{eq:Wk}, we have for any $s_0\in\sS$, $D(\{W,f,s_0\},\{W_k,f,s_0\}) = \frac{2}{k}$.
	Next, let us assume that $C_{FB}(\{W,f,s_0\})$, $s_0\in\{0,1\}$, is a continuous function on $\tco{\sP(\sY|\sX\times\sS)}\times\sS^{\sS\times\sX\times\sY}$. Then  we must have $\lim_{k\rightarrow\infty}C_{FB}(\{W_k,f,0\})=C_{FB}(\{W,f,0\})$ and $\lim_{k\rightarrow\infty}C_{FB}(\{W_k,f,1\})=C_{FB}(\{W,f,1\})$.
	Since for all $k\in\N$ the FSC $\{W_k,f,s_0\}$, $s_0\in\sS$, is indecomposable, we then  have $C_{FB}(\{W_k,f,0\})=C_{FB}(\{W_k,f,1\})$ and consequently obtain
	\begin{align*}
		1 &= C_{FB}(\{W,f,0\}) = \lim_{k\rightarrow\infty}C(\{W_k,f,0\})\\
		& = \lim_{k\rightarrow\infty}C(\{W_k,f,1\})= C_{FB}(\{W,f,1\})\\
		&=\log_2\Big(1+2^{-g(\epsilon)}\Big)<1
	\end{align*}
	with $g(\epsilon)=\frac{H_2(\epsilon)}{1-\epsilon}$ and $\epsilon\in(0,\frac{1}{2})\cap\Q$. This is a contradiction. Hence, at least one of the functions $C(\{W,f,0\})$ or $C(\{W,f,1\})$ must be discontinuous proving the desired result.
\end{proof}

\begin{thm}
	\label{thm:singleletter_cont}
	Let $|\sX|\geq2$, $|\sY|\geq2$, and $|\sS|\geq2$ be arbitrary. Then there is no natural number $n_0\in\N$ such that the capacity $C_{FB}(\{W,f,s_0\})$ can be expressed as
	\begin{equation}
		\label{eq:singleletter_cont}
		C_{FB}(\{W,f,s_0\}) = \max_{u\in\sU}F(u,W,f,s_0)
	\end{equation}
	with $\sU\subset\R^{n_0}$ a compact set and $F: \sU\times\tco{\sP(\sY|\sX\times\sS)}\times\sS^{\sS\times\sX\times\sY}\times\sS\rightarrow\R$ a continuous function.
\end{thm}
\begin{proof}
We use the same line of arguments as for \cite[Theorem ~1]{boche2019identification} and \cite{boche2020shannon}. The crucial observation is the following: To be able to express the capacity $C_{FB}(\{W,f,s_0\})$ as in \eqref{eq:singleletter_cont}, the capacity necessarily needs to be a continuous function. This cannot be the case shown by Theorem \ref{thm:discont}.
\end{proof}

Theorem \ref{thm:singleletter_cont} implies that the feedback capacity cannot be expressed by a finite multi-letter formula. This implies that there is no closed form solution possible in general for the capacity of FSCs with feedback.

\section{Conclusion}
\label{sec:conclusion}


In this paper we studied the capacity of FSCs with feedback from an algorithmic point of view. We showed that the feedback capacity function is not Banach-Mazur computable, \tco{which is} the weakest form of computability. Hence, the capacity of FSCs with feedback is \tco{also} not Borel-Turing computable. There are FSCs, for which the feedback capacity was computed, as shown in \cite{permuter2008capacity,elishco2014capacity,sabag2015feedback,sabag2015feedback2}. However, there is no general algorithm that takes $\{W,f,s_0\}$ as input and computes the feedback capacity $C_{FB}(\{W,f,s_0\})$. If the capacity of FSCs with feedback had been computable, then this could yield a solution for the halting problem, \tcr{which has been proven to be unsolvable.} 

Since the capacity of FSCs with feedback is not compuatble, one could aim to use upper bounds to compare the behavior of the coding procedures. For a meaningful evaluation, tight upper and lower bounds are desired. Unfortunately, we have further shown that we cannot find tight upper and lower bounds that are simultaneously computable. Meaning that either achievability or converse are non-computable. Hence, when developing coding procedures for the FSC with feedback, it makes hard to evaluate how is good the performance of that code.

The capacity of FSCs with feedback is a multi-letter formula, which means that it is a sequences of optimization problems. This makes it especially difficult to compute. Finding a finite letter formulation could facilitate computing the capacity. However, we show that the capacity of FSCs with feedback cannot be expressed as a finite multi-letter optimization problem. Hence, non of the approaches studied in this paper to approximate the capacity allow us to compute the capacity of FSCs with feedback.

However, it would be interesting to study if restricting the FSC set has an influence on the computability behavior of the feedback capacity.
 
\balance
\bibliographystyle{IEEEtran}
\bibliography{IEEEabrv,confs-jrnls,references_ISIT2022}

\end{document}